\RequirePackage{fix-cm}
\documentclass[smallcondensed]{svjour3}

\smartqed  

\usepackage{ucs}
\usepackage[utf8x]{inputenc}

\usepackage{mathptmx}  
\usepackage{helvet}    
\usepackage{textgreek} 

\usepackage{amsmath}
\usepackage{amssymb}

\usepackage{tikz-cd}
\usetikzlibrary{decorations.pathmorphing}
\usetikzlibrary{arrows}
\usetikzlibrary{matrix}

\tikzcdset{sdiag/.style={
    arrows={decorate,decoration={snake,amplitude=.3mm,segment length=1.5mm,post length=1mm}},
    labels={font=\small},
    column sep=large,
    row sep=large}}

\usepackage[bw]{agda}         
\AgdaNoSpaceAroundCode
\usepackage{mdframed}
\surroundwithmdframed[backgroundcolor=black!5!white,
                      innerleftmargin=5pt,
                      innertopmargin=10pt,
                      linecolor=white]{code}

\newcommand{\agda}[1]{\mbox{\small\textsf{#1}}}
\newcommand{\agdakw}[1]{\mbox{\small\textsf{\textbf{#1}}}}

\usepackage{hyperref}

\makeatletter
\def\operator@font{\sf}
\makeatother

\newcommand{\Type}{\mathsf{Type}}

\newcommand{\fst}{\operatorname{fst}}
\newcommand{\snd}{\operatorname{snd}}

\newcommand{\idpS}{\mathsf{idp}}
\newcommand{\idp}[1]{\idpS_{#1}}

\newcommand{\concat}{\mathbin{\cdot}}
\newcommand{\inv}{^{-1}}
\newcommand{\transport}{\operatorname{transport}}

\newcommand{\Unit}{\mathbf{1}}
\newcommand{\ttt}{\star_\Unit}

\newcommand{\N}{\mathbb{N}}

\newcommand{\Z}{\mathbb{Z}}

\newcommand{\apS}{\operatorname{ap}}
\newcommand{\ap}[1]{\apS_{#1}}
\newcommand{\apiiS}{\operatorname{ap}^2}
\newcommand{\apii}[1]{\apiiS_{#1}}
\newcommand{\defeq}{\mathrel{:=}}
\newcommand{\id}{\mathsf{id}}

\newcommand{\Sn}[1]{\mathbb{S}^{#1}}
\newcommand{\north}{\mathsf{north}}
\newcommand{\south}{\mathsf{south}}
\newcommand{\merid}{\operatorname{merid}}

\newcommand{\inn}{\operatorname{in}}
\newcommand{\push}{\operatorname{push}}

\newcommand{\inl}{\operatorname{inl}}
\newcommand{\inr}{\operatorname{inr}}

\newcommand{\fold}{\nabla}

\newcommand{\JiA}{J_\infty A}
\newcommand{\epsiloni}{\varepsilon_\infty}
\newcommand{\alphai}{\alpha_\infty}
\newcommand{\deltai}{\delta_\infty}
\newcommand{\gammai}{\gamma_\infty}
\newcommand{\etai}{\eta_\infty}

\newcommand{\innJ}{\operatorname{in}^J}
\newcommand{\pushJ}{\push^J}

\newcommand{\tof}{\operatorname{to}}
\newcommand{\from}{\operatorname{from}}

\newcommand{\pp}{\mathop{\hat{\times}}}

\newcommand{\apd}[1]{\operatorname{apd}_{#1}}

\newcommand{\Susp}{\upSigma}
\renewcommand{\Omega}{\upOmega}

\newenvironment{diagram}{%
\begin{figure}[htbp]
  \renewcommand{\figurename}{Diagram}
  \centering}%
{\end{figure}}

\newcommand{\epsilonJ}{\varepsilon_{J}}
\newcommand{\alphaJ}{\alpha_{J}}
\newcommand{\deltaJ}{\delta_{J}}
\newcommand{\gammaJ}{\gamma_{J}}
\newcommand{\etaJ}{\eta_{J}}

\newcommand{\urlgithub}{\url{https://github.com/guillaumebrunerie/JamesConstruction}}

\journalname{Journal of Automated Reasoning}

\begin{document}

\title{The James construction and $\pi_4(\Sn3)$ in homotopy type theory%
  \thanks{This material is based upon work supported by the National Science Foundation under
    agreement Nos. DMS-1128155 and CMU 1150129-338510.%
  }}
\author{Guillaume Brunerie}

\institute{Guillaume Brunerie \at
  Institute for Advanced Study, Princeton, NJ, USA \\
  \email{guillaume.brunerie@ias.edu}}

\date{Uploaded: 27th October 2017}

\maketitle

\begin{abstract}
  In the first part of this paper we present a formalization in Agda of the James construction in
  homotopy type theory. We include several fragments of code to show what the Agda code looks like,
  and we explain several techniques that we used in the formalization. In the second part, we use
  the James construction to give a constructive proof that $\pi_4(\Sn3)$ is of the form $\Z/n\Z$
  (but we do not compute the $n$ here).  
\end{abstract}

\section{Introduction}\label{sec:introduction}
\label{intro}

In this paper we define the James construction in homotopy type theory and we prove that
$\pi_4(\Sn3)$ is of the form $\Z/n\Z$. We have formalized the most technical part (the James
construction) in Agda and we present here numerous fragments of codes and remarks on the
formalization\footnote{The code is available at \urlgithub{} and has been tested with Agda
  2.5.2. The code fragments are generated directly from the source code using
  \mbox{\texttt{agda -{}-latex}} and a custom script to extract the relevant parts.}.
This article is based on chapter 3 of the author’s PhD thesis (\cite{phd}), but the formalization in
Agda of the James construction is new. In \cite{phd}, we also proved that $n$ is equal to $2$, but
this is out of the scope of the present paper.

The general idea of the James construction is that given a type $A$ pointed by $\star_A:A$, we
consider the higher inductive type $JA$ generated by the constructors
\begin{align*}
  \epsilonJ &: JA, \\
  \alphaJ &: A\to JA\to JA, \\
  \deltaJ &: (x:JA)\to x=_{JA}\alphaJ(\star_A, x).
\end{align*}
We can see $JA$ as the free monoid on $A$, where $\star_A$ is identified with the neutral
element. The function $\alphaJ$ builds sequences of elements of $A$ ($\epsilonJ$ corresponding to
the empty sequence), and the function $\deltaJ$ allows us to remove occurences of $\star_A$.

This higher inductive type is recursive, given that $JA$ itself appears in the domains of $\alphaJ$
and $\deltaJ$, and we would like to turn it into a non-recursive one. So we will define (in section
\ref{sec:JiA}) a sequence of types $(J_nA)_{n:\N}$ together with maps
$(i_n : J_nA\to J_{n+1}A)_{n:\N}$ such that the type $\JiA$, defined as the sequential colimit of
$(J_nA)_{n:\N}$, is equivalent to $JA$.

This equivalence between $JA$ and $\JiA$ is interesting because on the one hand we can show that
$JA$ is equivalent to $\Omega\Susp A$ when $A$ is connected (see section \ref{sec:equivJA}), and on
the other hand we can study the lower homotopy groups of $\JiA$ (see section
\ref{sec:connin}). Those two facts together allow us to prove that $\pi_4(\Sn3)$ is equal to
$\Z/n\Z$, where $n$ is the Whitehead product of the generator of $\pi_2(\Sn2)$ with itself (see
sections \ref{sec:wproducts} and \ref{sec:pi4s3}).

The main technical part of the paper is the proof that $JA$ and $\JiA$ are equivalent. The idea is
simple: we construct two functions going back and forth (in section \ref{sec:twofunctions}) and we
prove that they are inverse to each other (in section \ref{sec:twocomposites}). But $JA$ and $\JiA$
having quite different definitions, it requires careful manipulation of 2-dimensional and
3-dimensional diagrams.

Note that we have formalized only the James construction (sections \ref{sec:JiA} to
\ref{sec:twocomposites}), as it is the most technical part. We’re planning to formalize the rest in
the future, but it hasn’t been done at the time of this writing.

This definition of $JA$ and of the $(J_nA)_{n:\N}$ was suggested to me by André Joyal.

\section{Remarks on the formalization}\label{sec:formalization}

In this section we touch on three topics that are used extensively in the formalization: higher
inductive types using rewrite rules, cubical reasoning and coherence operations.

\subsection{Higher inductive types}

We start by recalling the definition of homotopy pushouts using higher inductive types (see for
instance chapter 6 of \cite{hottbook}), and we explain how we implemented them in Agda. As is usual
in homotopy type theory, we will call them simply “pushouts”, as this is the only sort of pushout of
types that we can define. Let’s consider three types $A$, $B$, $C$ and two functions $f:C\to A$,
$g:C\to B$,
\[
\begin{tikzcd}
  A & C \arrow[l,"f"'] \arrow[r,"g"] & B.
\end{tikzcd}\]
Such a diagram is called a \emph{span}. The \emph{pushout} of this span is the higher inductive type
$A\sqcup^CB$ generated by the constructors
\begin{align*}
\inl&:A\to A\sqcup^CB,\\
\inr&:B\to A\sqcup^CB,\\
\push&:(c:C)\to\inl(f(c))=_{A\sqcup^CB}\inr(g(c)).
\end{align*}
In particular, we have the square
\[
\begin{tikzcd}
  C \arrow[r,"g"] \arrow[d,"f"'] \arrow[rd, phantom, "\scriptstyle\push"] & B \arrow[d,dashed,"\inr"]\\
  A \arrow[r,dashed,"\inl"'] & A\sqcup^CB
\end{tikzcd}
\]
which is \emph{commutative}, in the sense that it is commutative up to identity paths~: the witness
of commutativity $\push$ is a pointwise path between the two functions corresponding to the two
compositions of the sides of the square.
The idea is that we start with the disjoint sum $A+B$, and for every element $c$ of $C$ we add a new
path from $\inl(f(c))$ to $\inr(g(c))$.

The induction principle states that, given a dependent type $P:A\sqcup^CB\to\Type$, we can define a
function $h:(x:A\sqcup^CB)\to P(x)$ by
\begin{align*}
  h &: (x:A\sqcup^CB)\to P(x),\\
  h(\inl(a)) &\defeq \inl^*(a),\\
  h(\inr(b)) &\defeq \inr^*(b),\\
  \apd h(\push(c)) &\defeq \push^*(c),
\end{align*}
where we have
\begin{align*}
  \inl^* &: (a : A) \to P(\inl(a)),\\
  \inr^* &: (b : B) \to P(\inr(b)),\\
  \push^* &: (c : C) \to \inl^*(f(c)) =^P_{\push(c)} \inr^*(g(c)).
\end{align*}
We are using here the notion of \emph{dependent paths} (see \cite{cubicalDan}): given a type $X$, a
dependent type $P:X\to\Type$, a path $p:x=x'$ in $X$ and two points $u:P(x)$ and $v:P(x')$, the type
\[u=^P_pv\] represents paths in $P$ going from $u$ to $v$ and lying over $p$. Given $h:(x:X)\to
Q(x)$ and $q:x=_Xx'$, the term $\apd{h}(q)$ is the application of $h$ to $q$, which is a dependent
path in $Q$, over $q$, and from $h(x)$ to $h(x')$.

We are using the same type theory as in \cite{hottbook}, in particular we take the first two
equalities (defining $h(\inl(a))$ and $h(\inr(b))$) to be judgmental equalities whereas the equality
between $\apd h(\push(c))$ and $\push^*(c)$ is only taken as a propositional equality.

In the formalization, higher inductive types are implemented using rewrite rules, which is an
experimental feature of Agda allowing the user to add (almost) arbitrary reduction rules to the type
theory, see \cite{rewrite}. It gives a cleaner implementation of higher inductive types than what
was used so far in both Agda and Coq (Dan Licata’s trick), as it doesn’t rely on declaring a fake
inductive type and inconsistent axioms and then trusting the hiding mechanism to only export the
part which is consistent. Here we simply postulate (i.e. introduce axioms for) the type, the
constructors, the elimination rule and the reduction rules, and then we tell Agda to treat the
reduction rules for points as judgmental equalities.

The corresponding Agda code is shown in fragment \ref{code:Pushout}. Here are some explanations to
help with the understanding:
\begin{itemize}
\item The variables \agda{i}, \agda{j} and \agda{k} are universe levels (they are declared at the
  top of the file, not shown here) and \agda{lsucc} and \agda{lmax} are operations of universe
  levels. Agda has explicit universe polymorphism and no cumulativity, which is why we need three
  different universe levels in order to have the most general notion of pushout.

\item The type \agda{Span} is defined as a record type with fields \agda{A}, \agda{B}, \agda{C},
  \agda{f} and \agda{g}. In order to construct a span, we use the syntax \agda{span A B C f g}
  (because we declared \agda{span} as the constructor), and given a span \agda{d}, the command
  \agda{\agdakw{open} Span d} brings the components \agda{A}, \agda{B}, \agda{C}, \agda{f} and
  \agda{g} of \agda{d} into scope.

\item We use the notion of \emph{anonymous module} (modules named "\agda{\_}"): the idea is simply
  to factor out common arguments of several definitions.

\item We use the notation \agda{u == v} for the identity type $u=v$, \agda{idp} for the identity
  path (also known as reflexivity), and \agda{u == v [ P ↓ p ]} for $u=^P_pv$. Dependent paths are
  implemented by induction on $p$, which means that the type $u=^P_{\idp{x}}v$ is equal to the type
  $u=_{P(x)}v$ by definition.

\item The rewriting mechanism works as follows. First Agda has to be started with the option
  \texttt{-{}-rewriting} (not shown here) to enable it. Then we declare the rewriting relation using
  the pragma \agda{\{-\# \agdakw{BUILTIN REWRITE} \_$\mapsto$\_ \#-\}}, see fragment
  \ref{code:rewrite}. Finally, we declare individual rewrite rules using \agda{\{-\#
    \agdakw{REWRITE} rew \#-\}}.

\item The reduction rule \agda{push-βd'} is primed simply because we usually want its arguments
  \agda{inl*}, \agda{inr*} and \agda{push*} to be implicit. We define \agda{push-βd} afterwards with
  those arguments made implicit (not shown here). Moreover, the \agda{d} at the end is there because
  we will also need the non-dependent reduction rule \agda{push-β}, which has a slightly different
  type.
\end{itemize}

\begin{figure}[htbp]
  \renewcommand{\figurename}{Code fragment}
  \centering
  \input{code/Pushout}
  \caption{The definition of pushouts}
  \label{code:Pushout}
\end{figure}

\begin{figure}[htbp]
  \renewcommand{\figurename}{Code fragment}
  \centering
  \input{code/rewrite}
  \caption{The type of rewrite rules}
  \label{code:rewrite}
\end{figure}

When \agda{P} is constant, we obtain the non-dependent elimination rule \agda{Pushout-rec}, see
fragment \ref{code:PushoutRec}. The function \agda{↓-cst-in} turns a homogeneous path into a
dependent path in the constant fibration, and the function \agda{apd=cst-in} turns an equality
\agda{apd f p == ↓-cst-in q} into the equality \agda{ap f p == q}, where \agda{f} is a non-dependent
function.

\begin{figure}[htbp]
  \renewcommand{\figurename}{Code fragment}
  \centering
  \input{code/PushoutRec}
  \caption{The non-dependent elimination rule and the associated reduction rule}
  \label{code:PushoutRec}
\end{figure}

We use the same scheme for all higher inductive types. For $JA$, the induction principle states that
given a dependent type $P:JA\to\Type$, a function $f:(x:JA)\to P(x)$ can be defined by
\begin{align*}
  f &: (x:JA)\to P(x),\\
  f(\epsilonJ) &\defeq \epsilonJ^*,\\
  f(\alphaJ(a,x)) &\defeq \alphaJ^*(a,x,f(x)),\\
  \apd{f}(\deltaJ(x)) &\defeq \deltaJ^*(x,f(x)),
\end{align*}
where we have
\begin{align*}
  \epsilonJ^*&:P(\epsilonJ),\\
  \alphaJ^*&:(a:A)(x:JA)\to P(x)\to P(\alphaJ(a,x)),\\
  \deltaJ^*&:(x:JA)(y:P(x)) \to y=^P_{\deltaJ(x)}\alphaJ^*(\star_A,x,y).
\end{align*}
Note that $f$ is used recursively in $f(\alphaJ(a,x))$ and in $\apd{f}(\deltaJ(x))$, because $JA$ is
a recursive higher inductive type. The code is shown in fragment \ref{code:defJA}.

\begin{figure}[htbp]
  \renewcommand{\figurename}{Code fragment}
  \centering
  \input{code/defJA}
  \caption{The definition of $JA$}
  \label{code:defJA}
\end{figure}

\subsection{Cubical reasoning}

In various places, we use cubical reasoning as in \cite{cubicalDan}. The main idea is that a
dependent path in an identity type
\[u=^{\lambda x.f(x)=g(x)}_pv\]
should be seen as a square
\[
\begin{tikzcd}[sdiag]
  f(x) \arrow[r,"\ap f(p)"] \arrow[d,"u"'] & f(x') \arrow[d,"v"]\\
  g(x) \arrow[r,"\ap g(p)"'] & g(x')
\end{tikzcd}
\]
In the formalization, the type of such squares is written \agda{Square u (ap f p) (ap g p) v},
i.e. we give the sides in the order left/top/bottom/right.  We use this idea in several
situations. One is when defining a function of type $(x:X)\to f(x)=_Yg(x)$ where $X$ is a higher
inductive type. If we use the elimination rule for $X$, for the path constructors we will need to
construct a dependent path in the dependent type $\lambda x.f(x)=_Yg(x)$, i.e. a square in
$Y$. Another situation is when we want to apply a function $h:(x:X)\to f(x)=_Yg(x)$ to a path
$p:x=x'$ in $X$. Using $\apd{}$ we obtain a dependent path in the dependent type above, so it
makes sense to see it as the following square (called the \emph{naturality square} of $h$ on $p$)
\[
\begin{tikzcd}[sdiag]
  f(x) \arrow[r,"\ap f(p)"] \arrow[d,"h(x)"'] & f(x') \arrow[d,"h(x')"]\\
  g(x) \arrow[r,"\ap g(p)"'] & g(x')
\end{tikzcd}
\]

There are similar results for cubes. In particular, a dependent path in a square type can be seen as
a cube, and similarly for a dependent square in a path type.


\subsection{Coherence operations}

We often have to compose together paths, squares, 2-dimensional paths, and so on, in a wide variety
of ways. Even though all such compositions can in theory be written using only a small number of
elementary operations, it is not always convenient to write them in such a way. We found that it is
often better to define ad-hoc operations on the fly. For instance, in section \ref{sec:twofunctions}
we need to define the composition of the following diagram, where $\nu$ and $\eta$ are squares
filling their respective part of the diagram and $vw^=$ is a 2-dimensional path between $vw$ and
$v\cdot w\inv$.
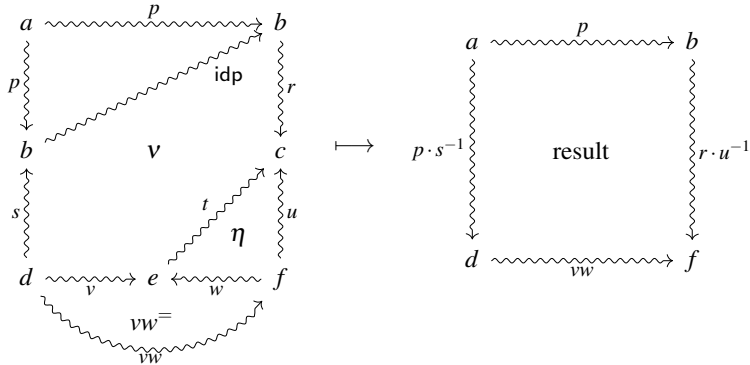
\[
\begin{tikzcd}[sdiag, row sep=huge]
  a \arrow[rr,"p"] \arrow[d,"p"'] & & b
  \arrow[d,"r"]\\
  b \arrow[rru,"\idpS"',near end] \arrow[rr,phantom,"\nu"] & & c\\
  d \arrow[u,"s"] \arrow[r,"v"'] \arrow[rr, bend right=50, "vw"'{name=VW}]
  & e \arrow[ru,"t"{name=T}] \arrow[to=VW, phantom, "vw^="] & f \arrow[l,"w"] \arrow[u,"u"'] \arrow[from=T, phantom, "\eta"]
\end{tikzcd}
\quad
\longmapsto
\quad
\begin{tikzcd}[sdiag, row sep=huge]
  a \arrow[rr,"p"] \arrow[dd,"p \cdot s\inv"'] \arrow[rrdd,phantom, "\text{result}"]& & b
  \arrow[dd,"r \cdot u\inv"]\\
  \\
  d \arrow[rr,"vw"'] & & f
\end{tikzcd}
\]

The key is to notice that the diagram is “contractible”, and that it is possible to write the list
of the arguments in a particular order reflecting this contractibility. More precisely, the
arguments are introduced in pairs \agda{(x : X) (y : Y)} where \agda{Y} is either an identity type
with \agda{x} as exactly one of the endpoints, or a square type with \agda{x} as one of the sides
(and not appearing in the other sides). We can then repeatedly apply the \agda{J} rule (or a similar
rule for squares) until the list of arguments is exhausted, and we finally return the identity
square.

We implemented a mechanism making it relatively easy to define such coherence operations in Agda. A
coherence operation is defined by encapsulating its type in the \agda{Coh} type constructor, and is
defined using the \agda{path-induction} term. See fragment \ref{code:excoh} for an application of
this principle to our example.

\begin{figure}[htbp]
  \renewcommand{\figurename}{Code fragment}
  \centering
  \input{code/excoh}
  \caption{Exemple of coherence operation}
  \label{code:excoh}
\end{figure}

This is implemented in Agda using instance arguments (the equivalent of type classes in Coq or
Haskell), see fragment \ref{code:pathinduction} for a simplified implementation. The type
constructor \agda{Coh} is a dummy record type which is used to make the instance arguments machinery
work. We then define the Paulin--Mohring rule \agda{J}, acting on terms in \agda{Coh}, and the
identity path under \agda{Coh}. Both \agda{J} and \agda{idp-Coh} are declared under the
\agdakw{instance} keyword, which means that whenever Agda is looking for an element of type
\agda{Coh} during instance resolution, it will automatically (and recursively) try both \agda{J} and
\agda{idp-Coh}. The term \agda{path-induction} then tells Agda to use the instance resolution
mechanism to try to solve the goal. For instance, in the term \agda{composition}, Agda is looking
for something of type
\begin{center}
  \agda{Coh (\{b : A\} (p : a == b) \{c : A\} (q : b == c) → a == c)}
\end{center}
It turns out that \agda{J} fits assuming you have something of type
\begin{center}
  \agda{Coh (\{c : A\} (q : a == c) → a == c)}
\end{center}
Again, \agda{J} fits assuming you have something of type
\begin{center}
  \agda{Coh (a == a)}
\end{center}
And in this case, \agda{idp-Coh} fits, so we are done. Therefore, the term \agda{path-induction}
simply reduces to \agda{J (J idp-Coh)}. For more complicated coherence operations, there might be
several $J$-like operators to be used, for instance if the path is reversed, or if we’re dealing
with squares, or if the arguments are implicit, but the user only has to type \agda{path-induction}
and instance resolution will automatically find the sequence of $J$-like operators to apply. The
resulting coherence operation can be turned into an actual function using the \agda{\&} function, as
is shown in \agda{pq}.

This mechanism can also be used to do inductions on homotopies (point-wise equality between
functions) or on equivalences (using the univalence axiom), for instance, by adding the appropriate
$J$-like operators.

\begin{figure}[htbp]
  \renewcommand{\figurename}{Code fragment}
  \centering
  \input{code/pathinduction}
  \caption{The \agda{\scriptsize path-induction} mechanism}
  \label{code:pathinduction}
\end{figure}

Note that there is a strong similarity between coherence operations as described here and operations
in a Grothendieck $\infty$-groupoid, the main difference being that we allow squares and other
shapes, whereas in a Grothendieck $\infty$-groupoid everything is strictly globular. In particular,
all operations in a Grothendieck $\infty$-groupoid are coherence operations as described here.

\section{Definition of the types $(J_nA)$ and $\JiA$}\label{sec:JiA}

We can now start working on the James construction. In this section we will define the types
$(J_nA)$ and $\JiA$. The intuition is that if $JA$ is the free monoid on $A$, then $J_nA$ is the
“subset” of $JA$ consisting of elements of length at most $n$. But this is only an intuition, as
there is no notion of “subset” which would apply here, and there is no notion of length for the
elements of $JA$ either, so we need to give a new definition.

The types $(J_nA)$ are defined by induction on $n$, together with three functions
\begin{align*}
  i_n &: J_nA\to J_{n+1}A,\\
  \alpha_n &: A\times J_nA\to J_{n+1}A,\\
  \beta_n &: (x:J_nA)\to\alpha_n(\star_A,x)=_{J_{n+1}A}i_n(x),
\end{align*}
as follows.
\begin{itemize}
\item $J_0A$ is the unit type, whose unique element is called $\epsilon$,
\item $J_1A\defeq A$, $i_0(\epsilon)\defeq\star_A$, $\alpha_0(a,\epsilon)\defeq a$ and
  $\beta_0(\epsilon)\defeq\idp{\star_A}$,
\item $J_{n+2}A$, $i_{n+1}$, $\alpha_{n+1}$ and $\beta_{n+1}\defeq\push\circ\inr$ are
  defined by the pushout diagram
  \begin{equation}
    \begin{tikzcd}
      (A\times J_nA)\sqcup^{J_nA}J_{n+1}A \arrow[r,"g"] \arrow[d,"f"'] & J_{n+1}A
      \ar[d,"i_{n+1}",dashed]\\
      A\times J_{n+1}A \ar[r,"\alpha_{n+1}"',dashed] & J_{n+2}A \arrow[lu,phantom,"\ulcorner",at start]
    \end{tikzcd}\label{eq:jn+2a}
  \end{equation}
  where the pushout at the top-left of the diagram is defined by the maps $x\mapsto(\star_A,x)$ and
  $i_n$, and the maps $f$ and $g$ are defined by
  \begin{align*}
    f(\inl(a,x))&\defeq(a, i_n(x)), & g(\inl(a,x))&\defeq\alpha_n(a,x),\\
    f(\inr(y))&\defeq (\star_A, y),     & g(\inr(y))&\defeq y,\\
    \ap f(\push(x))&\defeq \idpS,  & \ap g(\push(x))&\defeq \beta_n(x).
  \end{align*}
\end{itemize}
We could also have started the definition with $J_{-1}A$ being the empty type, and then it would
follow that $J_1A$ is equivalent to $A$, but we’ve decided to start at $J_0A$ so that we don’t need
to introduce negative numbers. Moreover, the data of $i_n$ and $\beta_n$ forms a contractible type,
as $\beta_n$ asserts that $i_n$ is equal to something else. Therefore, we could define $J_{n+2}A$ as
a higher inductive type using only $\alpha_n$, by simply substituting $\alpha_n(\star_A,x)$ for
$i_n(x)$ wherever needed. We decided to introduce $i_n$ and $\beta_n$ because defining $J_{n+2}A$ as
a pushout will be very helpful in order to get the connectivity results of section \ref{sec:connin}.

The Agda definition of the $J_nA$, $i_n$, $\alpha_n$ and $\beta_n$ is given in fragment
\ref{code:defJiab}. It is a set of mutually recursive definitions, which is written in Agda by
placing the type signatures of all the functions before their definitions. We write \agda{J n} for
$J_nA$ (the type $A$ being a global argument), \agda{JS n} for $J_{n+1}A$ (we need to define it
separately in order to pass the termination checker), \agda{ι n x} for $i_n(x)$, \agda{α n a x} for
$\alpha_n(a, x)$ and \agda{β n x} for $\beta_n(x)$.

\begin{figure}[htbp]
  \renewcommand{\figurename}{Code fragment}
  \centering
  \input{code/defJiab}
  \caption{The definition of $J_nA$, $i_n$, $\alpha_n$ and $\beta_n$}
  \label{code:defJiab}
\end{figure}

Note that $J_{n+2}A$ is defined by giving $i_{n+1}$, $\alpha_{n+1}$, $\beta_{n+1}$, and the two
functions
\begin{align*}
  \gamma_n &: (a : A) (x : J_nA) \to \alpha_{n+1}(a,i_n(x))=_{J_{n+2}A}i_{n+1}(\alpha_n(a,x)),\\
  \gamma_n(a,x) &\defeq \push(\inl(a,x))
\end{align*}
and
\[
\eta_n : (x : J_nA) \to
\begin{tikzcd}[sdiag]
  \bullet \arrow[r,"\idp{}"] \arrow[d,"\gamma_n({\star_A,x})"'] & \bullet
  \arrow[d,"\beta_{n+1}(i_n(x))"]\\
  \bullet \arrow[r,"\ap{i_{n+1}}(\beta_n(x))"'] & \bullet
\end{tikzcd}
\]
which is the naturality square of $\push$ on $\push(x)$. 


We could also have defined $J_{n+2}A$ directly as a higher inductive type with constructors
$i_{n+1}$, $\alpha_{n+1}$, $\beta_{n+1}$, $\gamma_n$ and $\eta_n$. But in section \ref{sec:connin}
we will use the fact that it is defined using pushouts, so instead we simply prove that $J_{n+2}A$
satisfies the elimination rule corresponding to those five constructors. This will be very useful
when defining functions out of $J_{n+2}A$. The code is shown in fragment \ref{code:JSSElim}. Note
that we need to use a dependent square over $\eta_n(x)$, given that $\eta_n(x)$ is a square. The
function \agda{↓-ap-in} turns a dependent path in $P \circ i_n$ over $\beta_n(x)$ into a dependent
path in $P$ over $\ap{i_n}(\beta_n(x))$, and the function \agda{↓-ap-in-coh} is a coherence related
to \agda{↓-ap-in}. The important thing to see is that we use twice the elimination rule for
pushouts, and that we put $\iota^*$, $\alpha^*$, $\beta^*$, $\gamma^*$ and $\eta^*$ in the five
branches, which is what we should expect.

\begin{figure}[htbp]
  \renewcommand{\figurename}{Code fragment}
  \centering
  \input{code/JSSElim}
  \caption{The elimination rule of $J_{n+2}A$}
  \label{code:JSSElim}
\end{figure}


We now define $\JiA$ as the colimit of the family $(J_nA)_{n:\N}$ along the maps $(i_n)_{n:\N}$,
which means that $\JiA$ is the higher inductive type generated by the two constructors
\begin{align*}
  \inn&:(n:\N)\to J_nA\to \JiA,\\
  \push&:(n:\N)(x:J_nA)\to\inn_n(x)=_{\JiA}\inn_{n+1}(i_n(x)).
\end{align*}
The induction principle for $\JiA$ states that given a dependent type $P:\JiA\to\Type$, a
function $f:(x:\JiA)\to P(x)$ can be defined by
\begin{align*}
  f &: (x:\JiA)\to P(x),\\
  f(\inn_n(x)) &\defeq \inn^*_n(x),\\
  \apd{f}(\push_n(x)) &\defeq \push^*_n(x),
\end{align*}
where we have
\begin{align*}
  \inn^*_n&:(x:J_nA)\to P(\inn_n(x)),\\
  \push^*_n&:(x:J_nA)\to \inn^*_n(x) =^P_{\push_n(x)}\inn^*_{n+1}(i_n(x)).
\end{align*}
It is implemented is the same way as for pushouts and $JA$, and the corresponding code is shown in
fragment \ref{code:defJiA}. Note that we’re using the notations \agda{in∞ n x} for $\inn_n(x)$ and
\agda{push∞ n x} for $\push_n(x)$, because \agdakw{in} is a reserved keyword in Agda and \agda{push}
is already used for pushouts.

\begin{figure}[htbp]
  \renewcommand{\figurename}{Code fragment}
  \centering
  \input{code/defJiA}
  \caption{The definition of $\JiA$}
  \label{code:defJiA}
\end{figure}


\section{The two functions}\label{sec:twofunctions}

We recall that $JA$ is the higher inductive type with constructors
\begin{align*}
  \epsilonJ &: JA, \\
  \alphaJ &: A\to JA\to JA, \\
  \deltaJ &: (x:JA)\to x=_{JA}\alphaJ(\star_A, x).
\end{align*}
In this section we define the two maps between $JA$ and $\JiA$. The idea is to mimic the structure
present in $JA$ in $\JiA$, and vice versa, so we first define equivalents of $\gamma_n$, $\eta_n$,
$\inn_n$ and $\push_n$ in $JA$, and then equivalents of $\epsilonJ$, $\alphaJ$, $\deltaJ$, and of
$\gammaJ$ and $\etaJ$ (defined below) in $\JiA$.

\paragraph{Structure on $JA$}

We define the map $\gammaJ$, where we simply apply $\deltaJ$ twice, by
\begin{align*}
  \gammaJ &: (a : A) (x : JA) \to \alphaJ(a,\alphaJ(\star_A,x)) =
                  \alphaJ(\star_A,\alphaJ(a, x)),\\
  \gammaJ(a,x) &\defeq (\ap{\alphaJ(a,-)}{(\deltaJ(x))})\inv\concat\deltaJ(\alphaJ(a,x)),
\end{align*}
and the map
\begin{align*}
  \etaJ &: (x : JA) \to \gammaJ(\star_A, x) = \idpS
\end{align*}
using naturality of $\deltaJ$ on $\deltaJ(x)$, see diagram \ref{diag:natdeltadelta}.

\begin{diagram}
  \begin{tikzcd}[sdiag,column sep=huge]
    x \arrow[d,"{\deltaJ(x)}"'] \arrow[r,"\deltaJ(x)"] &
    \alphaJ(\star_A,x) \arrow[d,"{\deltaJ(\alphaJ(\star_A,x))}"]\\
    \alphaJ(\star_A,x) \arrow[r,"\ap{\alphaJ(\star_A,-)}(\deltaJ(x))"'] &
    \alphaJ(\star_A,\alphaJ(\star_A, x))
  \end{tikzcd}
  
  \caption{Naturality square of $\deltaJ$ on $\deltaJ(x)$}
\label{diag:natdeltadelta}
\end{diagram}

The formalization of $\gammaJ$ and $\etaJ$ is shown in fragment \ref{code:defgJeJ}. We will define
$\gammai$ and $\etai$ in the same way, which is why we wrote it for a general type $X$ equipped with
functions $\alpha$ and $\delta$. The coherence operation shows that given a square where the left
and top sides are the same, then the inverse of the bottom side composed with the right side is
equal to the identity path, which is what we need when defining $\etaJ$.

\begin{figure}[htbp]
  \renewcommand{\figurename}{Code fragment}
  \centering
  \input{code/defgJeJ}
  \caption{The definition of $\gammaJ$ and $\etaJ$}
  \label{code:defgJeJ}
\end{figure}

We now define $(\innJ_n)$ and $(\pushJ_n)$ by
\begin{align*}
  \innJ_n &: J_nA \to JA, &  \pushJ_n &: (x : J_nA) \to \innJ_n(x) = \innJ_{n+1}(i_n(x)),\\
  \innJ_0(\epsilon) &\defeq \epsilonJ, &  \pushJ_n(x) &\defeq \deltaJ(\innJ_n(x)).\\
  \innJ_1(a) &\defeq \alphaJ(a,\epsilonJ),&&\\
  \innJ_{n+2}(i_{n+1}(x)) &\defeq \alphaJ(\star_A,\innJ_{n+1}(x)),&&\\
  \innJ_{n+2}(\alpha_{n+1}(a,x)) &\defeq \alphaJ(a,\innJ_{n+1}(x)),&&\\
  \ap{\innJ_{n+2}}{(\beta_{n+1}(x))} &\defeq \idpS,&&\\
  \ap{\innJ_{n+2}}{(\gamma_n(a,x))} &\defeq \gammaJ(a,\innJ_n(x)),&&\\
  \apii{\innJ_{n+2}}{(\eta_n(x))} &\defeq \etaJ(\innJ_n(x)),&&
\end{align*}
Note that for $\innJ_{n+2}$ we’re using the new (non-dependent) elimination rule for $J_{n+2}A$
mentioned earlier.  While this definition looks simple a priori, it doesn’t quite type-check. In
particular, the type of the term $\ap{\innJ_{n+2}}{(\gamma_n(a,x))}$ is
\[\innJ_{n+2}(\alpha_{n+1}(a,i_n(x)))=\innJ_{n+2}(i_{n+1}(\alpha_n(a,x))),\]
whereas the type of $\gammaJ(a,\innJ_n(x))$ is
\[\alphaJ(a,\alphaJ(\star_A,\innJ_n(x))) = \alphaJ(\star_A,\alphaJ(a,\innJ_n(x))).\]
Looking at the definitions above, the outer $\innJ_{n+2}$ reduce, but then we need the following
reduction rules:
\begin{align*}
  \innJ_{n+1}(i_n(x)) &= \alphaJ(\star_A,\innJ_n(x)),\\
  \innJ_{n+1}(\alpha_n(a, x)) &= \alphaJ(a,\innJ_n(x)).
\end{align*}
The idea is that we have defined $\innJ_{n+1}(i_n(x))$ separately for $0$ and $n+1$, and we need to
make sure that the two are compatible (and similarly for $\alpha_n(x)$).  It is easy to see that
those equalities both hold definitionally when $n$ is either \agda{0} or of the form \agda{S n}, but
that does \emph{not} imply that they hold definitionally for an arbitrary $n$.

Therefore, in the formalization we use propositional equalities \agda{inJS-ι} and \agda{inJS-α} that
we prove together with the rest, and in the definition of $\ap{\innJ_{n+2}}(\gamma_n(a,x))$ we need
to explicitly compose $\gammaJ(a,\innJ_n(x))$ with those equalities. For
$\ap{\innJ_{n+2}}(\eta_n(x))$ we also need a similar equality corresponding to $\beta_n(x)$. The
code is shown in fragment \ref{code:definJpushJ}.

\begin{figure}[htbp]
  \renewcommand{\figurename}{Code fragment}
  \centering
  \input{code/definJpushJ}
  \caption{The definition of $\innJ_n$ and $\pushJ_n$}
  \label{code:definJpushJ}
\end{figure}

\paragraph{Structure on $\JiA$}

The equivalent of $\epsilonJ$ is the term $\epsiloni\defeq\inn_0(\epsilon)$ of type $\JiA$. We
then define the action of $A$ on $\JiA$ as follows. In order to multiply by $a:A$ an element of the
form $\inn_n(x)$, we use $\alpha_n$, and then we use $\gamma_n$ to show that it is compatible with
$i_n$.
\begin{align*}
  \alphai &: A \to \JiA \to \JiA,\\
  \alphai(a,\inn_n(x)) &\defeq \inn_{n+1}(\alpha_n(a,x)),\\
  \ap{\alphai(a,-)}{(\push_n(x))} &\defeq \push_{n+1}(\alpha_n(a,x))\cdot\ap{\inn_{n+2}}{(\gamma_n(a,x))}\inv.
\end{align*}
The equivalent of $\deltaJ$ is $\deltai$ defined by
\begin{align*}
  \deltai &: (x : \JiA) \to x=\alphai(\star_A,x), \\
  \deltai(\inn_n(x)) &\defeq \push_n(x) \concat \ap{\inn_{n+1}}(\beta_n(x))\inv, \\
  \apd\deltai(\push_n(x)) &\defeq \deltai^{\push_n}(x),
\end{align*}
where $\deltai^{\push_n}(x)$ is the composition of diagram \ref{diag:deltaipush}, where the lower
right triangle is filled using $\eta_n(x)$ and the pentagon in the middle is filled using the
naturality square of $\push_{n+1}$ on $\beta_n(x)$. The corresponding code is shown in fragment
\ref{code:defstructinf}.

\begin{figure}[htbp]
  \renewcommand{\figurename}{Code fragment}
  \centering
  \input{code/defstructinf}
  \caption{The structure on $\JiA$}
  \label{code:defstructinf}
\end{figure}

\begin{diagram}
\begin{tikzcd}[sdiag,sep=6.5em]
  \bullet \arrow[rr,"\push_n(x)"] \arrow[d,"\push_n(x)"'] & & \bullet
  \arrow[d,"\push_{n+1}(i_n(x))"]\\
  \bullet \arrow[rru,"\idpS"'] & & \bullet\\
  \bullet \arrow[u,"\ap{\inn_{n+1}}(\beta_n(x))"] \arrow[r,"{\push_{n+1}(\alpha_n(\star_A,x))}"']
  & \bullet \arrow[ru,"\ap{\inn_{n+2}}(\ap{i_{n+1}}(\beta_n(x)))"] & \bullet \arrow[l,"{\ap{\inn_{n+2}}(\gamma_n(\star_A,x))}"] \arrow[u,"\ap{\inn_{n+2}}(\beta_{n+1}(i_n(x)))"']
\end{tikzcd}
\caption{The square defining $\apd\deltai(\push_n(x))$}
\label{diag:deltaipush}
\end{diagram}

We finally define 
\begin{align*}
  \gammai &: (a : A) (x : \JiA) \to \alphai(a,\alphai(\star_A,x)) =
                  \alphai(\star_A,\alphai(a, x)),\\
  \etai &: (x : \JiA) \to \gammai(\star_A, x)=\idpS
\end{align*}
in the same way as we defined $\gammaJ$ and $\etaJ$, but using $\alphai$ and $\deltai$ instead of
$\alphaJ$ and $\deltaJ$.

In the case of $\gammai(a,\inn_n(x))$ we note that
\begin{align*}
  \gammai(a,\inn_n(x))
  &= \ap{\alphai(a,-)}(\deltai(\inn_n(x)))\inv\concat\deltai(\alphai(a,\inn_n(x)))\\
  &= \ap{\alphai(a,-)}(\push_n(x)\concat
    \ap{\inn_{n+1}}(\beta_n(x))\inv)\inv\concat\deltai(\inn_{n+1}(\alpha_n(a,x)))\\
  &= (\push_{n+1}(\alpha_n(a,x))\concat\ap{\inn_{n+2}}(\gamma_n(a,x))\inv\\
  &\qquad\concat\ap{\inn_{n+2}}(\ap{\alpha_{n+1}(a,-)}(\beta_n(x)))\inv)\inv\\
  &\quad\concat(\push_{n+1}(\alpha_n(a,x))\concat\ap{\inn_{n+2}}(\beta_{n+1}(\alpha_n(a,x)))\inv)\\
  &= \ap{\inn_{n+2}}(\ap{\alpha_{n+1}(a,-)}(\beta_n(x)))\concat\ap{\inn_{n+2}}(\gamma_n(a,x))\\
  &\qquad\concat\ap{\inn_{n+2}}(\beta_{n+1}(\alpha_n(a,x)))\inv.
\end{align*}
Therefore $\gammai(a,\inn_n(x))$ fits in the square
\begin{equation}
\begin{tikzcd}[sdiag,sep=huge]
  \bullet \arrow[r,"{\ap{\inn_{n+2}}(\ap{\alpha_{n+1}(a,-)}(\beta_n(x)))}"]
  \arrow[d,"{\gammai(a,\inn_n(x))}"'] & \bullet \arrow[d,"{\ap{\inn_{n+2}}(\gamma_n(a,x))}"]\\
  \bullet \arrow[r,"{\ap{\inn_{n+2}}(\beta_{n+1}(\alpha_n(a,x)))}"'] & \bullet
\end{tikzcd}\label{eq:gammai}
\end{equation}
which we can see as a sort of reduction rule for $\gammai(a,\inn_n(x))$. In the formalization, we
simply define a coherence operation combining all the ingredients of the equality reasoning above,
see fragment \ref{code:gammainfin}.

\begin{figure}[htbp]
  \renewcommand{\figurename}{Code fragment}
  \centering
  \input{code/gammainfin}
  \caption{The reduction rule for $\gammai(a,\inn_n(x))$}
  \label{code:gammainfin}
\end{figure}

There is a similar reduction rule for $\etai(\inn_n(x))$. The term $\etai(\inn_n(x))$ is defined
using $\apd{\deltai}(\deltai(\inn_n(x)))$ and we have
\begin{align*}
  \apd{\deltai}(\deltai(\inn_n(x))) &= \apd{\deltai}(\push_n(x)\concat
                                     \ap{\inn_{n+1}}(\beta_n(x)\inv))\\
                                   &= \deltai^{\push_n}(x)\concat\apd{\lambda x.\push_{n+1}(x)\concat\ap{\inn_{n+2}}(\beta_{n+1}(x)\inv)}(\beta_n(x)\inv).
\end{align*}
The $\apd{\push_{n+1}}(\beta_n(x)\inv)$ part cancels with the naturality square of $\push_{n+1}$ on
$\beta_n(x)$ used in $\deltai^{\push_n}(x)$ and the remaining part
$\apd{\ap{\inn_{n+2}}(\beta_{n+1}(-)\inv)}(\beta_n(x)\inv)$ is the naturality square of $\beta_{n+1}$
on $\beta_n(x)$. Therefore $\etai(\inn_n(x))$ fits in the three-dimensional diagram
\begin{equation}
\begin{tikzcd}[sdiag,sep=6em,row sep=normal]
  \bullet \arrow[r,"{\ap{\inn_{n+2}}(\ap{\alpha_{n+1}(\star_A,-)}(\beta_n(x)))}"]
  \arrow[dd,"{\gammai(\star_A,\inn_n(x))}"' description] \arrow[dd,bend right=90,looseness=2,"\idpS"']& \bullet
  \arrow[dd,"{\ap{\inn_{n+2}}(\gamma_n(\star_A,x))}" description] \arrow[rrd,"{\ap{\inn_{n+2}}(\beta_{n+1}(i_n(x)))}"]\\
  & & & \bullet\\
  \bullet \arrow[r,"{\ap{\inn_{n+2}}(\beta_{n+1}(\alpha_n(\star_A,x)))}"'] & \bullet \arrow[rru,"\ap{\inn_{n+2}}(\ap{i_{n+1}}(\beta_n(x)))"']
\end{tikzcd}
\label{diagetai}
\end{equation}
where the half-disc on the left is $\etai(\inn_n(x))$, the square in the middle is square
(\ref{eq:gammai}), the triangle on the right is the application of $\inn_{n+2}$ to $\eta_n(x)$ and the
outer diagram is the application of $\inn_{n+2}$ to the naturality square of $\beta_{n+1}$ on
$\beta_n(x)$, which is
\[
\begin{tikzcd}[sdiag,sep=huge]
  \bullet \arrow[r,"{\ap{\alpha_{n+1}(\star_A,-)}(\beta_n(x))}"]
  \arrow[d,"{\beta_{n+1}(\alpha_n(\star_A,x))}"'] & \bullet \arrow[d,"{\beta_{n+1}(i_n(x))}"]\\
  \bullet \arrow[r,"{\ap{i_{n+1}}(\beta_n(x))}"'] & \bullet
\end{tikzcd}
\]
As we see it as a reduction rule for $\etai(\inn_n(x))$, in the formalization it is helpful to see
it as a cube, where the left face is $\etai(\inn_n(x))$, the right face is $\ap{\inn_n}(\eta_n(x))$,
and the other faces are what is needed to make the sides of the left and right face coincide.  As
before, it is defined as a coherence operation combining all the ingredients described above.



\paragraph{The two maps}

We can now define the maps back and forth by
\begin{align*}
  \tof &: \JiA \to JA, & \from &: JA \to \JiA,\\
  \tof(\inn_n(x)) &\defeq \innJ_n(x), & \from(\epsilonJ) &\defeq \epsiloni,\\
  \ap{\tof}{(\push_n(x))} &\defeq \pushJ_n(x), & \from(\alphaJ(a,x)) &\defeq \alphai(a,\from(x)),\\
  & & \ap{\from}(\deltaJ(x)) &\defeq \deltai(\from(x)).
\end{align*}
The code, given in fragment \ref{code:twomaps}, is straightforward.

\begin{figure}[htbp]
  \renewcommand{\figurename}{Code fragment}
  \centering
  \input{code/twomaps}
  \caption{The two maps}
  \label{code:twomaps}
\end{figure}

\section{The two composites}\label{sec:twocomposites}

We now prove that the two maps \agda{to} and \agda{from} are inverse to each other. We will stop
giving code fragments, as they would become too long, but we remind the reader that the full code is
available at \urlgithub.

\paragraph{First composite}

Let’s first prove that $\from(\tof(z)) = z$ for all $z:\JiA$.

By induction on $z$, it is enough to show that for every $n:\N$ and $x:J_nA$, we have
$\from(\innJ_n(x))=\inn_n(x)$ and $\ap\from(\pushJ_n(x))=\push_n(x)$ (in the appropriate dependent
path type). Let’s first do the case of $\innJ_n(x)$ by induction on $n$, and then by induction on
$x$, using the dependent elimination rule for $J_{n+2}A$.
\begin{itemize}
\item For $0$ and $\epsilon$, we have
  \begin{align*}
    \from(\innJ_0(\epsilon)) &= \from(\epsilonJ)\\
                             &= \epsiloni\\
    &= \inn_0(\epsilon).
  \end{align*}
\item For $1$ and $a:A$, we have
  \begin{align*}
    \from(\innJ_1(a)) &= \from(\alphaJ(a,\epsilonJ))\\
                      &= \alphai(a,\from(\epsilonJ))\\
                      &= \inn_1(a).
  \end{align*}
\item For $n+2$ and $i_{n+1}(x)$, we have
  \begin{align*}
    \from(\innJ_{n+2}(i_{n+1}(x))) &= \from(\alphaJ(\star_A,\innJ_{n+1}(x)))\\
                                   &= \alphai(\star_A,\from(\innJ_{n+1}(x)))\\
                                   &= \alphai(\star_A,\inn_{n+1}(x))\quad\text{ by induction hypothesis}\\
                                   &= \inn_{n+2}(\alpha_{n+1}(\star_A,x))\\
                                   &= \inn_{n+2}(i_{n+1}(x))\quad \text{ using }\beta_{n+1}(x).
  \end{align*}
\item For $n+2$ and $\alpha_{n+1}(a,x)$, we have
  \begin{align*}
    \from(\innJ_{n+2}(\alpha_{n+1}(a,x))) &= \from(\alphaJ(a,\innJ_{n+1}(x)))\\
                                   &= \alphai(a,\from(\innJ_{n+1}(x)))\\
                                   &= \alphai(a,\inn_{n+1}(x))\quad\text{ by induction hypothesis}\\
                                   &= \inn_{n+2}(\alpha_{n+1}(a,x)).
  \end{align*}
\item For $n+2$ and $\beta_{n+1}(x)$, we have
  \begin{align*}
    \ap\from(\ap{\innJ_{n+2}}(\beta_{n+1}(x))) &= \ap\from(\idp{\alphaJ(\star_A,\innJ_{n+1}(x))})\\
                                               &= \idp{\from(\alphaJ(\star_A,\innJ_{n+1}(x)))}\\
                                               &= \idp{\alphai(\star_A,\from(\innJ_{n+1}(x)))}
  \end{align*}
  hence it follows from the fact that the diagram
  \[
  \begin{tikzcd}[sdiag,column sep=8em]
    \bullet \arrow[r,"p"] \arrow[d,"{\idp{\alphai(\star_A,\from(\innJ_{n+1}(x)))}}"'] &
    \bullet \arrow[r,"{\idp{\inn_{n+2}(\alpha_{n+1}(\star_A,x))}}"]
    \arrow[d,bend right,"{\idp{\alphai(\star_A,\inn_{n+1}(x))}}"']
    \arrow[d,bend left,"{\idp{\inn_{n+2}(\alpha_{n+1}(\star_A,x))}}"] &
    \bullet \arrow[d,"\ap{\inn_{n+2}}(\beta_{n+1}(x))"] \\
    \bullet \arrow[r,"p"'] &
    \bullet \arrow[r,"\ap{\inn_{n+2}}(\beta_{n+1}(x))"'] & \bullet
  \end{tikzcd}
  \]
  can be filled. Here the path
  $p:\alphai(\star_A,\from(\innJ_{n+1}(x)))=\alphai(\star_A,\inn_{n+1}(x))$ is the function
  $\alphai(\star_A,-)$ applied to the induction hypothesis, the two curved paths in the middle are
  definitionally equal, and the right square is a connection. The top and the bottom side are the
  equalities between $\from(\innJ_{n+2}(x))$ and $\inn_{n+2}(x)$ constructed above for
  $x\defeq\alpha_{n+1}(\star_A,x)$ and $x\defeq i_{n+1}(x)$, which is what we want.
\item For $n+2$ and $\gamma_n(a,x)$, we need to give a square
  \[
  \begin{tikzcd}[sdiag]
    \bullet \arrow[r] \arrow[d,"\ap\from(\ap{\innJ_{n+2}}(\gamma_n({a,x})))"'] & \bullet \arrow[d,"\ap{\inn_{n+2}}(\gamma_n({a,x}))"] \\
    \bullet \arrow[r] & \bullet
  \end{tikzcd}
  \]
  where the top and bottom lines are the two equalities
  \begin{align*}
    \from(\innJ_{n+2}(\alpha_{n+1}(a,i_n(x)))) &= \inn_{n+2}(\alpha_{n+1}(a,i_n(x)))\\
    \text{and}\quad\from(\innJ_{n+2}(i_{n+1}(\alpha_n(a,x)))) &= \inn_{n+2}(i_{n+1}(\alpha_n(a,x)))
  \end{align*}
  which have been constructed in the cases above. The idea is to consider the following sequence of
  equalities
  \begin{alignat*}{2}
    \ap\from(\ap{\innJ_{n+2}}(\gamma_n(a,x))) &= \ap\from(\gammaJ(a,\innJ_n(x)))&&\quad\text{ by
      definition of $\innJ$,}\\
    &= \gammai(a,\from(\innJ_n(x)))&&\quad\text{ by definition of $\from$,}\\
    &= \gammai(a,\inn_n(x))&&\quad\text{ by induction hypothesis,}\\
    &= \ap{\inn_{n+2}}(\gamma_n(a,x))&&\quad\text{ by diagram \ref{eq:gammai}.}
  \end{alignat*}
  The first, third and fourth of those equalities are actually squares, so the above equational
  reasoning means that we consider a composition of squares as follows:
  
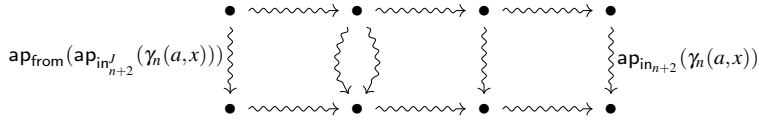
\[
  \begin{tikzcd}[sdiag]
    \bullet \arrow[r] \arrow[d,"\ap\from(\ap{\innJ_{n+2}}(\gamma_n({a,x})))"'] & \bullet \arrow[d, bend right] \arrow[d, bend left] \arrow[r]
    & \bullet \arrow[d] \arrow[r] & \bullet \arrow[d,"\ap{\inn_{n+2}}(\gamma_n({a,x}))"] \\
    \bullet \arrow[r] & \bullet \arrow[r] & \bullet \arrow[r] & \bullet
  \end{tikzcd}
  \]

  However, it turns out that the top and the bottom line of that composition of squares are
  \emph{not} definitionally equal to what we want. For instance the top lines both go from
  $\from(\innJ_{n+2}(\alpha_{n+1}(a,i_n(x))))$ to $\inn_{n+2}(\alpha_{n+1}(a,i_n(x)))$, but in two
  different ways, and we have to prove that they are equal. This isn’t complicated, but it needs to
  be done, and it’s not a priori obvious to see when just looking at the equational reasoning
  above.
  
\item For $n+2$ and $\eta_n(x)$, it is similar to the case of $\gamma_n$. The core of the argument
  is the sequence of equalities
  \begin{alignat*}{2}
    \ap\from^2(\ap{\innJ_{n+2}}^2(\eta_n(x))) &= \ap\from^2(\etaJ(\innJ_n(x)))&&\quad\text{ by
      definition of $\innJ$,} \\
    &= \etai(\from(\innJ_n(x)))&&\quad\text{ by definition of $\from$,} \\
    &= \etai(\inn_n(x))&&\quad\text{ by induction
      hypothesis,}\\
    &= \ap{\inn_{n+2}}^2(\eta_n(x))&&\quad\text{ by diagram \ref{diagetai},}
  \end{alignat*}
  but the terms involved are squares which do not always have the same sides, therefore in the
  formalization we need to consider a composition of cubes, and then as above we need to prove that
  all four faces are equal to the ones required by the elimination rule of $J_{n+2}A$, which isn’t a
  priori true.
\end{itemize}

We finally have to show that for every $n:\N$ and $x:J_nA$, we have an equality between
$\ap\from(\pushJ_n(x))$ and $\push_n(x)$ along the equalities \[\from(\innJ_n(x))=\inn_n(x)\] and
\[\from(\innJ_{n+1}(i_n(x)))=\inn_{n+1}(i_n(x))\] that we have just constructed. We have
\begin{align*}
  \ap\from(\pushJ_n(x)) &= \deltai(\from(\innJ_n(x))) \\
                        &= \deltai(\inn_n(x))\quad\text{ by induction hypothesis}\\
                        &= \push_n(x)\concat\ap{\inn_{n+1}}(\beta_n(x))\inv,
\end{align*}
hence we have a filler of the square
\[
\begin{tikzcd}[sdiag,sep=huge]
  \bullet \arrow[rr,"p"] \arrow[d,"\ap\from(\pushJ_n(x))"'] & & \bullet \arrow[d,"\push_n(x)"]\\
  \bullet \arrow[r,"\ap{\alphai(\star_A,-)}(p)"'] & \bullet \arrow[r,"\ap{\inn_{n+1}}(\beta_n(x))"'] & \bullet
\end{tikzcd}
\]
where $p$ is the equality $\from(\innJ_n(x))=\inn_n(x)$.

\paragraph{Second composite}

Let’s now prove that $\tof(\from(z)) = z$, for all $z:JA$.

The idea is very similar, we proceed by induction on $z$ and we have to prove that
$\tof(\epsiloni)=\epsilonJ$ (which is true by definition), that for all $a:A$ and $x:\JiA$ we have
$\tof(\alphai(a,x))=\alphaJ(a,\tof(x))$, and that for all $x:\JiA$ we have
$\ap\tof(\deltai(x))=\deltaJ(\tof(x))$ along the appropriate equalities. Let’s first do the case of
$\alphai$ by induction on $x$. There are two cases.
\begin{itemize}
\item For an element of the form $\inn_n(x)$, we have
  \begin{align*}
    \tof(\alphai(a,\inn_n(x))) &= \tof(\inn_{n+1}(\alpha_n(a,x)))\\
                               &= \alphaJ(a,\innJ_n(x))\\
                               &= \alphaJ(a,\tof(\inn_n(x))),
  \end{align*}
  which is what we wanted.
\item For a path of the form $\push_n(x)$, we have
  \begin{align*}
    \ap\tof(\ap{\alphai(a,-)}(\push_n(x))) &= \ap\tof(\push_{n+1}(\alpha_n(a,x))\concat
                                              \ap{\inn_{n+2}}(\gamma_n(a,x))\inv)\\
                                            &= \deltaJ(\innJ_{n+1}(\alpha_n(a,x)))
                                              \concat\gammaJ(a,\innJ_n(x))\inv\\
                                            &= \deltaJ(\alphaJ(a,\innJ_n(x)))
                                              \concat\gammaJ(a,\innJ_n(x))\inv\\
                                            &= \ap{\alphaJ(a,-)}(\deltaJ(\innJ_n(x)))\quad\text{ by
                                              definition of $\gammaJ$}\\
                                            &= \ap{\alphaJ(a,-)}(\pushJ_n(x))\\
                                            &= \ap{\alphaJ(a,-)}(\ap\tof(\push_n(x))),
  \end{align*}
  which is again what we wanted.
\end{itemize}

We now prove that the path $\ap\tof(\deltai(x)):\tof(x)=\tof(\alphai(\star_A,x))$ composed with the
path from $\tof(\alphai(\star_A,x))$ to $\alphaJ(\star_A,\tof(x))$ that we have just constructed is
equal to the path $\deltaJ(\tof(x)):\tof(x)=\alphaJ(\star_A,\tof(x))$.
\begin{itemize}
\item For an element of the form $\inn_n(x)$, we have
  \begin{align*}
    \ap\tof(\deltai(\inn_n(x))) &= \ap\tof(\push_n(x)\concat\ap{\inn_{n+1}}(\beta_n(x))\inv)\\
                                &= \pushJ_n(x)\\
                                &= \deltaJ(\innJ_n(x))\\
                                &= \deltaJ(\tof(\inn_n(x))),
  \end{align*}
  which proves the result, as the path from $\tof(\alphai(\star_A,\inn_n(x)))$ to
  $\alphaJ(\star_A,\tof(\inn_n(x)))$ is the constant path.
\item For a path of the form $\push_n(x)$, we have to compare the terms
  $\ap\tof^2(\apd\deltai(\push_n(x)))$ and $\apd\deltaJ(\ap\tof(\push_n(x)))$. For the first one, we
  just apply the function $\tof$ to diagram \ref{diag:deltaipush}. We obtain
  
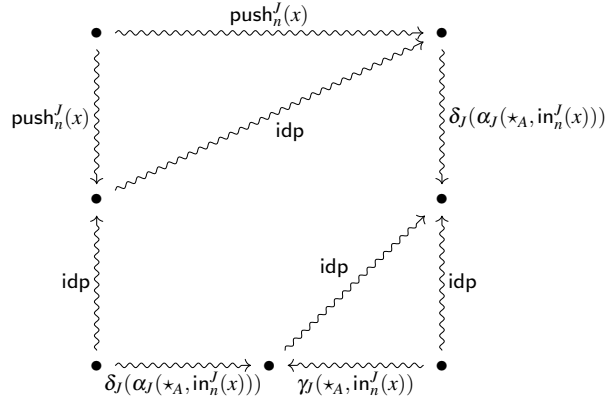
\[
  \begin{tikzcd}[sdiag,sep=5.4em]
    \bullet \arrow[rr,"\pushJ_n(x)"] \arrow[d,"\pushJ_n(x)"'] & & \bullet
    \arrow[d,"{\deltaJ(\alphaJ(\star_A,\innJ_n(x)))}"]\\
    \bullet \arrow[rru,"\idpS"'] & & \bullet\\
    \bullet \arrow[u,"\idpS"] \arrow[r,"{\deltaJ(\alphaJ(\star_A,\innJ_n(x)))}"']
    & \bullet \arrow[ru,"\idpS"] & \bullet \arrow[l,"{\gammaJ(\star_A,\innJ_n(x))}"] \arrow[u,"\idpS"']
  \end{tikzcd}
  \]
  where the bottom right triangle is filled using $\etaJ(\innJ_n(x))$ and the rest is degenerate.
  On the other hand, we have
  \begin{align*}
    \apd\deltaJ(\ap\tof(\push_n(x))) &= \apd\deltaJ(\pushJ_n(x))\\
    &= \apd\deltaJ(\deltaJ(\innJ_n(x)))
  \end{align*}
  and $\etaJ(\innJ_n(x))$ is defined from $\apd\deltaJ(\deltaJ(\innJ_n(x)))$ by a coherence
  operation. Therefore some coherence operation proves that the two terms
  $\ap\tof^2(\apd\deltai(\push_n(x)))$ and $\apd\deltaJ(\ap\tof(\push_n(x)))$ are equal.
\end{itemize}

This concludes the proof that $\JiA$ is equivalent to $JA$.

\section{Equivalence between $JA$ and $\Omega\Susp A$}\label{sec:equivJA}

We now prove that when $A$ is $0$-connected, the type $JA$ is equivalent to $\Omega\Susp A$. We
recall that $\Omega X$ is defined to be $(\star_X=\star_X)$, where $X$ is a type pointed by
$\star_X:X$, and that given a type $A$, the type $\Susp A$ is the higher inductive type generated by
the constructors
\begin{align*}
  \north&:\Susp A,\\
  \south&:\Susp A,\\
  \merid&:A\to \north=_{\Susp A}\south,
\end{align*}
and pointed by $\north$.  The function $\deltaJ$ shows that the map $\alphaJ(\star_A, -)$ is homotopic
to the identity function, hence $\alphaJ(\star_A,-)$ is an equivalence. Given that $A$ is
$0$-connected, it follows that $\alphaJ(a,-)$ is an equivalence for every $a:A$. We define $F:\Susp
A\to\Type$ by
\begin{align*}
  F(\north) &\defeq JA,\\
  F(\south) &\defeq JA,\\
  \ap F(\merid(a)) &\defeq \mathsf{ua}(\alphaJ(a,-)),
\end{align*}
where, at the last line, the function $\mathsf{ua}$ produces a path in the universe given an
equivalence, using the univalence
axiom.

We now prove that the total space of $F$ is contractible. According to the flattening lemma (see
\cite[section 6.12]{hottbook}), the total space of $F$ is equivalent to the type
\[T\defeq JA\sqcup^{A\times JA}JA,\]
where the two maps $A\times JA\to JA$ are $\snd$ and $\alphaJ$ respectively. In particular, given
$a:A$ and $x:JA$, we have
\[\push(a,x) : \inl(x) = \inr(\alphaJ(a,x)).\]
We want to construct, for every $x:T$, a path $p(x)$ from $x$ to $\inl(\epsilonJ)$.
\begin{itemize}
\item For $\inl(\epsilonJ)$, we take the constant path $\idp{\inl(\epsilonJ)}$
\item For an element of the form $\inl(\alphaJ(a,x))$, we take the composition
  \[
  \begin{tikzcd}[sdiag]
    \inl(\alphaJ(a,x)) \arrow[rrr,"{\push(\star_A,\alphaJ(a,x))}"] &&& \inr(\alphaJ(\star_A,\alphaJ(a,x)))
    \\
    \inl(x) \arrow[d,"p(\inl(x))"'] \arrow[rrr,"{\push(a,x)}"] &&&\inr(\alphaJ(a,x)) \arrow[u,"{\ap\inr(\deltaJ(\alphaJ(a,x)))}"']\\
    \inl(\epsilonJ).
  \end{tikzcd}
  \]
\item For a path of the form $\ap\inl(\deltaJ(x))$, we need to fill the diagram
  \[
  \begin{tikzcd}[sdiag]
    \inl(\alphaJ(\star_A,x)) \arrow[rrr,"{\push(\star_A,\alphaJ(\star_A,x))}"]
     &&& \inr(\alphaJ(\star_A,\alphaJ(\star_A,x))) \\
    \inl(x) \arrow[u,"\ap\inl(\deltaJ(x))"] \arrow[rrr,"{\push(\star_A,x)}"'] &&& \inr(\alphaJ(\star_A,x))
    \arrow[u,"{\ap\inr(\deltaJ(\alphaJ(\star_A,x)))}"']
  \end{tikzcd}
  \]
  By naturality of $\push(\star_A,-)$ on the path $\deltaJ(x)$, we get a filler of the similar
  diagram which has $\ap\inr(\ap{\alphaJ(\star_A,-)}(\deltaJ(x)))$ on the right side. Moreover, we
  know that the paths $\ap\inr(\ap{\alphaJ(\star_A,-)}(\deltaJ(x)))$ and
  $\ap\inr(\deltaJ(\alphaJ(\star_A,x)))$ are equal via $\etaJ(x)$, which concludes.
\item For a point of the form $\inr(x)$, we take the composition
  \[
  \begin{tikzcd}[sdiag]
    \inr(x) \arrow[r,"\ap\inr(\deltaJ(x))"]& \inr(\alphaJ(\star_A,x))&&
    \inl(x) \arrow[rr,"p(\inl(x))"] \arrow[ll,"{\push(\star_A,x)}"'] &&\inl(\epsilonJ).
  \end{tikzcd}
  \]
\item Finally for a path of the form $\push(a,x)$, it is enough to notice that the path from
  $\inr(\alphaJ(\star_A,x))$ to $\inl(\epsilonJ)$ constructed above (i.e. with
  $x:=\alphaJ(\star_A,x)$) is equal to the composition
  \[
  \begin{tikzcd}[sdiag]
    \inr(\alphaJ(a,x)) && \inl(x) \arrow[r,"p(\inl(x))"] \arrow[ll,"{\push(a,x)}"'] &\inl(\epsilonJ).
  \end{tikzcd}
  \]
\end{itemize}

This concludes the proof that $T$ is contractible, and therefore $\Omega\Susp A$ is equivalent to
the fiber $F(\north)$ of $F$ (it follows from Theorem 4.7.7 and Corollary 8.1.13 of
\cite{hottbook} applied to $F$), which is equal to $JA$ by definition.

\section{Connectivity of the maps $i_n$ and $\inn_n$}\label{sec:connin}

In this section we compute the connectivity of the maps $\inn_n:J_nA\to \JiA$. It quantifies how
“close” $J_nA$ is to $\JiA$, so it will be useful to study the first few homotopy groups of $\JiA$
by studying those of $J_nA$ instead.

We recall that a type $X$ is said to be \emph{$n$-connected} if its $n$-truncation is contractible,
and a map $f:X\to Y$ is said to be \emph{$n$-connected} if all of its (homotopy) fibers are
$n$-connected. We also recall the two following propositions.

\begin{proposition}\label{inductionconnected}[cf \cite[lemma 7.5.7]{hottbook}]
  For $f:A\to B$ and $P:B\to \Type$, consider the map
  \[\lambda s.s\circ f:\prod_{b:B}P(b) \to \prod_{a:A}P(f(a)).\]

  Then $f$ is $n$-connected if and only if for every family of $n$-types $P$, the map $(\lambda
  s.s\circ f)$ has a section.
\end{proposition}

\begin{proposition}\label{inductionconnectedtruncated}[cf \cite[lemma 8.6.1]{hottbook}]
  If $f:A \to B$ is $n$-connected and $P:B\to\Type$ is a family of $(n+k)$-types, then the map
  \[\lambda s.s\circ f:\prod_{b:B}P(b) \to \prod_{a:A}P(f(a))\] is $(k-2)$-truncated (in the sense that
  all its fibers are $(k-2)$-truncated).
\end{proposition}

Given two maps $f:X\to A$ and $g:Y\to B$, the \emph{pushout-product} of $f$ and $g$ is the map
\begin{align*}
  f\pp g &: (X\times B)\sqcup^{X\times Y}(A\times Y)\to (A\times B),\\
  (f\pp g)(\inl(x,b)) &\defeq (f(x),b),\\
  (f\pp g)(\inr(a,y)) &\defeq (a,g(y)),\\
  \ap{f\pp g}(\push(x,y)) &\defeq \idp{(f(x),g(y))}.
\end{align*}
We have the commutative square
\[
\begin{tikzcd}
  X\times Y \arrow[rr,"f\times1_Y"] \arrow[dd,"1_X\times g"'] && A\times Y \arrow[dd,"1_A\times g"]
  \arrow[ld,"\inr"'] \\
  & (X\times B)\sqcup^{X\times Y}(A\times Y) \arrow[rd,"f\pp g",dashed]
  \arrow[lu,phantom,"\ulcorner",at start] & \\
  X\times B \arrow[rr,"f\times1_B"] \arrow[ru,"\inl"] && A \times B
\end{tikzcd}
\]

We have the following proposition.
\begin{proposition}\label{pushoutproduct}
  If $f$ is $m$-connected and $g$ is $n$-connected, then $f\pp g$ is $(m+n+2)$-connected.
\end{proposition}

\begin{proof}

  We use proposition \ref{inductionconnected}. We consider $P:A\times B\to\Type$ a family of
  $(n+m+2)$-types together with
  \[k:(u:(X\times B)\sqcup^{X\times Y}(A\times Y)) \to P((f\pp g)(u)).\] By splitting $k$ in three
  parts and currying, it is enough to prove the following lemma.
  
  \begin{lemma}
    Suppose we have $P : A \to B \to \Type$ a family of $(n+m+2)$-types together with
    \begin{align*}
      u &: (x:X)(b:B)\to P(f(x),b),\\
      v &: (a:A)(y:Y)\to P(a, g(y)),\\
      w &: (x:X)(y:Y)\to u(x,g(y))=_{P(f(x),g(y))}v(f(x),y).
    \end{align*}
    Then there exists a map
    \[h : (a:A)(b:B)\to P(a,b)\] together with homotopies
    \begin{align*}
      p &: (x:X)(b:B)\to h(f(x),b) = u(x,b),\\
      q &: (a:A)(y:Y)\to h(a,g(y)) = v(a,y),\\
      r &: (x:X)(y:Y)\to p(x,g(y))\inv\cdot q(f(x),y) = w(x,y).
    \end{align*}
  \end{lemma}
  
  \begin{proof}
    Let’s define $F:A\to\Type$ by
    \[F(a)\defeq \sum_{k:(b:B) \to P(a,b)}((y:Y) \to k(g(y)) = v(a,y)).\]
    For a given $a:A$, the type $F(a)$ is the fiber of the map
    \[\lambda s.s\circ g:\prod_{b:B}P(a,b) \to \prod_{y:Y}P(a,g(y))\]
    at $v(a,-)$.
    Given that $g$ is $n$-connected and that $P(a,-)$ is a family of $(n+m+2)$-truncated types,
    proposition \ref{inductionconnectedtruncated} shows that $F(a)$ is $m$-truncated.
  
    For every $x:X$ we have an element of $F(f(x))$ given by $(u(x, -), w(x, -))$.  Hence, using the
    fact that $f$ is $m$-connected and proposition \ref{inductionconnected}, there is a map
    $k:(a:A)\to F(a)$ together with a homotopy $\varphi$ between $k\circ f$ and $\lambda
    x.(u(x,-),w(x,-))$.
    We can now define $h$, $p$, $q$, and $r$ by
    \begin{align*}
      h(a,b) &\defeq \fst(k(a))(b),\\
      p(x,b) &\defeq \fst(\varphi(x))(b),\\
      q(a,y) &\defeq \snd(k(a))(y),\\
      r(x,y) &\defeq \snd(\varphi(x))(y). 
    \end{align*}
  \end{proof}
  This concludes the proof that $f\pp g$ is $(n+m+2)$-connected.\qed
\end{proof}

We can now compute the connectivity of the maps $i_n$.

\begin{proposition}
  If $A$ is $k$-connected, then the map $i_n$ is $(n(k+1)+(k-1))$-connected for every $n:\N$.
\end{proposition}

\begin{proof}
  We proceed by induction on $n$. For $0$, the map $i_0$ is the inclusion of the basepoint of $A$,
  hence $i_0$ is $(k-1)$-connected because $A$ is $k$-connected.

  For $n+1$, the map $f$ in the diagram defining $J_{n+2}A$ (page \pageref{eq:jn+2a}) is the
  pushout-product of $i_n$ and of the map $\Unit\to A$ (which is $(k-1)$-connected). Hence $f$ is
  $((n+1)(k+1)+(k-1))$-connected by proposition \ref{pushoutproduct}. Therefore the map $i_{n+1}$ is
  $((n+1)(k+1)+(k-1))$-connected as well, because a pushout of an $\ell$-connected map is
  $\ell$-connected.\qed
\end{proof}

In the following proposition, we consider an arbitrary family of types $(A_n)_{n:\N}$ and maps
$(i_n:A_n\to A_{n+1})_{n:\N}$, with sequential colimit $A_\infty$.
\begin{proposition}
  Given $k:\N$, if all the maps $i_0$, $i_1$, \dots are $k$-connected, then $\inn_0$ is also
  $k$-connected.
\end{proposition}

\begin{proof}
  Let’s consider $P:A_\infty\to\Type$ a family of $k$-truncated types and $d_0:(x:A_0)\to
  P(\inn_0(x))$. Using proposition \ref{inductionconnected}, it is enough to construct a section $d$
  of $P$ which is equal to $d_0$ on $A_0$ to conclude that $\inn_0$ is $k$-connected. We define a
  family of maps $d_n:(x:A_n)\to P(\inn_n(x))$ by induction on $n$, starting with the given $d_0$
  for $n=0$, as follows. Let’s consider
  \begin{align*}
    P_{n+1} &: A_{n+1}\to\Type,\\
    P_{n+1}(x) &\defeq P(\inn_{n+1}(x)).
  \end{align*}
  It is a family of $k$-truncated types, the map $i_n$ is $k$-connected, and we have
  \begin{align*}
    \widetilde{d}_n &: (x : A_n) \to P_{n+1}(i_n(x)),\\
    \widetilde{d}_n(x) &\defeq \transport^P(\push_n(x),d_n(x)),
  \end{align*}
  therefore, using proposition \ref{inductionconnected} again, there is a map
  $d_{n+1}:(x:A_{n+1})\to P(\inn_{n+1}(x))$ satisfying
  \[
  d_{n+1}(i_n(x)) =^P_{\push_n(x)} d_n(x).
  \]
  The family $(d_n)_{n:\N}$ together with those equalities gives a section of $P$ which is equal to
  $d_0$ on $A_0$. Therefore, the map $\inn_0$ is $k$-connected, which is what we wanted to prove.\qed
\end{proof}

It follows that if the maps $i_n$, $i_{n+1}$, \dots are $k$ connected, then $\inn_n$ is also
$k$-connected. Therefore, in the case of the James construction, we have the following proposition.

\begin{proposition}
  If $A$ is $k$-connected, then the map $\inn_n:J_nA\to \JiA$ is $(n(k+1)+(k-1))$-connected for every
  $n:\N$.
\end{proposition}

Combining this result with those of the previous sections, for $n=1$ we obtain the Freudenthal
suspension theorem (a more direct proof was given in \cite[section 8.6]{hottbook}).

\begin{corollary}[Freudenthal suspension theorem]\label{freud}
  Given a $k$-connected pointed type $A$, the map
  \begin{align*}
    \varphi_A&:A\to\Omega\Susp A,\\
    \varphi_A(x) &:= \merid(x) \concat\merid(\star_A)\inv,
  \end{align*}
  is $2k$-connected.
\end{corollary}

For $n=2$, we obtain the following corollary.

\begin{corollary}\label{corjames}
  Given a $k$-connected pointed type $A$, there is a $(3k+1)$-connected map
  \[(A\times A)\sqcup^{A\vee A}A\to\Omega\Susp A,\]
\end{corollary}
where the \emph{wedge sum} $A\vee B$ of two pointed types $A$ and $B$ is defined as the pushout of
the span
\[
\begin{tikzcd}
  A & \Unit \arrow[l] \arrow[r] & B.
\end{tikzcd}
\]

Note that both corollaries are also true in the case $k=-1$ because every map is $(-2)$-connected.

\section{Whitehead products}\label{sec:wproducts}

In proposition \ref{whiteheadmap} we give a decomposition of a product of spheres into a pushout of
spheres. This will allow us to define Whitehead products, which are used in the next section to
define the natural number $n$ such that $\pi_4(\Sn3)\simeq\Z/n\Z$.

Given two pointed types $A$ and $B$, their \emph{join} $A*B$ is defined as the pushout of the span
\[
\begin{tikzcd}
  A & A\times B \arrow[l,"\fst"'] \arrow[r,"\snd"] & B.
\end{tikzcd}
\]
If $A$ and $B$ are spheres, one can show that we have the equivalence
\[\Sn{n-1}*\Sn{m-1} \simeq \Sn{n+m-1}.\]

\begin{proposition}\label{whiteheadmap}
  Given $n,m:\N^*$, there is a map $W_{n,m}:\Sn{n+m-1}\to\Sn{n}\vee\Sn{m}$ such that
  \[\Sn n\times\Sn m \simeq \Unit\sqcup^{\Sn{n+m-1}}(\Sn n\vee\Sn m),\]
  and such that the induced map $\Sn n\vee\Sn m\to\Sn n\times\Sn m$ is the canonical one.
\end{proposition}
We first prove the following more general version which isn’t more complicated to prove.
\begin{proposition}\label{whiteheadab}
  Given two types $A$ and $B$, there is a map $W_{A,B}:A*B\to\Susp A\vee\Susp B$ such that
  \[\Susp A\times\Susp B \simeq \Unit\sqcup^{A*B}(\Susp A\vee\Susp B)\]
  and such that the induced map $\Susp A\vee\Susp B\to\Susp A\times\Susp B$ is the canonical one.
\end{proposition}
\begin{proof}
  We consider the following diagram
  \[
  \begin{tikzcd}
    \Susp A & B \arrow[l,"\north"'] \arrow[r] \arrow[ld,"\alpha"',Rightarrow,shorten >= 1.5ex,
    shorten <= 1.5ex] & \Unit \\
    B \arrow[u,"\south"] \arrow[d,"\id"'] & A\times B \arrow[l,"\snd"] \arrow[r,"\fst"']
    \arrow[u,"\snd"']
    \arrow[d,"\snd"] & A \arrow[u] \arrow[d] \\
    B & B \arrow[l,"\id"] \arrow[r] & \Unit
  \end{tikzcd}
  \]
  where $\alpha:A\times B\to\north=_{\Susp A}\south$ is defined by $\alpha(x,y)\defeq\merid(x)$, and
  we use the $3\times3$-lemma (cf section VII of \cite{cubicalDan}) which states that the pushout of
  the pushouts of the rows is equivalent to the pushout of the pushouts of the columns.

  The pushout of the top row is equivalent to $\Susp A\vee\Susp B$, the pushout of the middle row is
  equivalent to the join $A*B$ and the pushout of the bottom row is contractible, so the pushout of
  the pushouts of the rows is equivalent to $\Unit\sqcup^{A*B}(\Susp A\vee\Susp B)$ for the map
  $W_{A,B}:A*B\to\Susp A\vee\Susp B$ defined by
  \begin{align*}
    W_{A,B} &: A*B \to \Susp A\vee\Susp B,\\
    W_{A,B}(\inl(a)) &\defeq \inr(\north),\\
    W_{A,B}(\inr(b)) &\defeq \inl(\north),\\
    \ap{W_{A,B}}(\push(a,b)) &\defeq \ap\inr(\varphi_B(b))\concat\push(\ttt)\concat\ap\inl(\varphi_A(a)).
  \end{align*}

  The pushouts of the left and of the right columns are both equivalent to $\Susp A$, and the
  pushout of the middle column is equivalent to $\Susp A\times B$. Moreover, the horizontal map on
  the left between $\Susp A\times B$ and $\Susp A$ is equal to $\fst$, as can be proved by
  induction using the definition of $\alpha$. The horizontal map on the right is also equal to
  $\fst$. Hence the pushout of the pushout of the columns is equivalent to $\Susp A\times\Susp
  B$. Therefore we have
  \[\Susp A\times\Susp B \simeq \Unit\sqcup^{A*B}(\Susp A\vee\Susp B)\]
  and it can be checked that the induced map $\Susp A\vee\Susp B\to\Susp A\times\Susp B$ is the
  canonical one. \qed
\end{proof}

\begin{proof}[of proposition \ref{whiteheadmap}]
  We apply proposition \ref{whiteheadab} to $A\defeq\Sn{n-1}$ and $B\defeq\Sn{m-1}$, and we obtain
  \[\Sn n\times\Sn m \simeq \Unit\sqcup^{\Sn{n-1}*\Sn{m-1}}(\Sn n\vee\Sn m).\]
  Moreover, we have $\Sn{n-1}*\Sn{m-1}\simeq\Sn{n+m-1}$, as mentioned earlier, which concludes. \qed
\end{proof}

This allows us to define the following operation on homotopy groups.
\begin{definition}
  Given a pointed type $X$ and two positive integers $n$ and $m$, the \emph{Whitehead product} is
  the function
  \[[-,-] : \pi_n(X) \times \pi_m(X) \to \pi_{n+m-1}(X)\]
  defined by composition with $W_{n,m}$ when representing elements of homotopy groups as maps from
  the spheres.
\end{definition}

\section{Application to homotopy groups of spheres}\label{sec:pi4s3}

The sphere $\Sn n$ is $(n-1)$-connected, therefore by the Freudenthal suspension theorem (corollary
\ref{freud}), the map $\varphi_{\Sn n}:\Sn n\to\Omega\Sn{n+1}$ is $(2n-2)$-connected. On homotopy
groups it gives the following result.

\begin{proposition}
  For $k,n:\N$, the map $\pi_{n+k}(\Sn n)\to\pi_{n+k+1}(\Sn{n+1})$ is an isomorphism if $n\ge k+2$
  and surjective if $n=k+1$.
\end{proposition}
This means that the groups $\pi_{n+k}(\Sn n)$ (for a fixed $k$) stabilize for a large enough $n$.
In particular, for $k=1$ we have the following result.
\begin{corollary}\label{pinsn}
  For $n\ge3$ we have $\pi_{n+1}(\Sn n)\simeq\pi_4(\Sn3)$ and the map $\pi_3(\Sn2)\to\pi_4(\Sn3)$ is
  surjective.
\end{corollary}

Note that even though we know that $\pi_3(\Sn2)\simeq\Z$ (from the Hopf fibration), as we are
working constructively this does \emph{not} imply that $\pi_4(\Sn3)$ is of the form $\Z/n\Z$ for
some $n:\N$. Indeed, it cannot be proved constructively that every subgroup of $\Z$ is of the form
$n\Z$, as there is no way to compute this $n$ in general. In this case, however, we can use the
James construction to give an explicit definition of the kernel of that map. We will need the
Blakers--Massey theorem (see \cite{blakersmassey}):

\begin{proposition}[Blakers--Massey theorem]
  Given two maps $f:C\to A$ and $g:C\to B$, we consider the types $D\defeq A\sqcup^CB$,
  \[E\defeq\sum_{a:A}\sum_{b:B}(\inl(a)=_D\inr(b)),\]
  and the map $h : C\to E$ defined by $h(c) \defeq (f(c),g(c),\push(c))$.
  \[
  \begin{tikzcd}
    C \arrow[rrd,bend left,"g"] \arrow[rdd,bend right,"f"'] \arrow[rd,dashed,"h"] & & \\
    & E \arrow[r] \arrow[d] \arrow[rd,phantom,"\lrcorner",at start] & B \arrow[d,"\inr"]\\
    & A \arrow[r,"\inl"'] & D
  \end{tikzcd}
  \]
  If $f$ is $n$-connected and $g$ is $m$-connected, then $h$ is $(n+m)$-connected.
\end{proposition}

We now prove the following proposition.
\begin{proposition}\label{kerneljames}
  For $n\ge2$, the kernel of the surjective map $\pi_{2n-1}(\Sn n)\to \pi_{2n}(\Sn{n+1})$ induced by
  $\varphi_{\Sn n}$ is generated by the Whitehead product $[i_n,i_n]$, where $i_n$ is the generator
  of $\pi_n(\Sn n)$.
\end{proposition}

\begin{proof}
  Applying corollary \ref{corjames} to $\Sn n$ which is $(n-1)$-connected, we get a
  $(3n-2)$-connected map from $J_2(\Sn n)$ to $\Omega\Sn{n+1}$. In particular, given that
  $2n-1<3n-2$, it means that
  \[\pi_{2n-1}(J_2(\Sn n))\simeq\pi_{2n-1}(\Omega\Sn{n+1})\simeq\pi_{2n}(\Sn{n+1}),\] so we now
  study the map $\pi_{2n-1}(\Sn n)\to \pi_{2n-1}(J_2(\Sn n))$.  We know from the James construction
  that \[J_2(\Sn n)\simeq(\Sn n\times\Sn n)\sqcup^{\Sn n\vee\Sn n}\Sn n,\]
  hence using the decomposition of $\Sn n\times\Sn n$ given in proposition \ref{whiteheadmap}, we
  get
  \[J_2(\Sn n)\simeq(\Unit\sqcup^{\Sn{2n-1}}(\Sn n\vee\Sn n))\sqcup^{\Sn n\vee\Sn n}\Sn n\]
  where the map from $\Sn n\vee\Sn n$ to the pushout on the left is $\inr$ (i.e.\ it’s the identity
  on the second component). This reduces to
  \[J_2(\Sn n)\simeq\Unit\sqcup^{\Sn{2n-1}}\Sn n,\]
  where the map $\Sn{2n-1}\to\Sn n$ is the Whitehead map $W_{n,n}:\Sn{2n-1}\to\Sn n\vee\Sn n$
  composed with the folding map $\fold_{\Sn n}:\Sn n\vee\Sn n\to\Sn n$.

  We now take the fiber $P$ of the map $\Sn n\to J_2(\Sn n)$, which is the pullback of the two
  maps from $\Sn n$ and $\Unit$ to $J_2(\Sn n)$
  \[
  \begin{tikzcd}
    \Sn{2n-1} \arrow[rrr,"{\fold_{\Sn n}\circ W_{n,n}}"] \arrow[ddd] \arrow[rd, dashed] &&& \Sn n \arrow[ddd] \\
    & P \arrow[rru] \arrow[ddl] \arrow[rrdd,phantom,"\lrcorner",at start] && \\ \\
    \Unit \arrow[rrr] &&& J_2(\Sn n) \arrow[lluu,phantom,"\ulcorner",at start]
  \end{tikzcd}
  \]
  The map from $\Sn{2n-1}$ to $\Unit$ is $(2n-2)$-connected and the map from $\Sn{2n-1}$ to $\Sn n$
  is $(n-2)$-connected (indeed, every map between two $(n-1)$-connected types is $(n-2)$-connected),
  hence using the Blakers-Massey theorem, the dashed map from $\Sn{2n-1}$ to $P$ is
  $(3n-4)$-connected. Given that $2n-2\le3n-4$, it follows that
  $\pi_{2n-2}(P)\simeq\pi_{2n-2}(\Sn{2n-1})\simeq0$.

  The long exact sequence of homotopy groups for $P\to \Sn n\to J_2(\Sn n)$ is
  \[
  \begin{tikzcd}
    \pi_{2n-1}(P) \arrow[r]& \pi_{2n-1}(\Sn n) \arrow[r]& \pi_{2n-1}(J_2(\Sn n)) \arrow[r]& \pi_{2n-2}(P)=0,
  \end{tikzcd}
  \]
  therefore $\pi_{2n-1}(J_2(\Sn n))$ is the quotient of $\pi_{2n-1}(\Sn n)$ by the image of the map
  $\pi_{2n-1}(P)\to\pi_{2n-1}(\Sn n)$. But the dashed map is surjective on $\pi_{2n-1}$, so it’s the
  same as the image of the map $\pi_{2n-1}(\Sn{2n-1})\to\pi_{2n-1}(\Sn n)$, which is generated by
  $[i_n,i_n]$, by definition of the Whitehead product.

  Therefore, the kernel of the map $\pi_{2n-1}(\Sn n)\to\pi_{2n}(\Sn{n+1})$ is generated by
  $[i_n,i_n]$. \qed
\end{proof}
In particular, applying this result to $n=2$ and using the fact that $\pi_3(\Sn2)\simeq\Z$, we get the
following corollary.
\begin{corollary}\label{firstpi4s3}
  We have
  \[\pi_4(\Sn3)\simeq\Z/n\Z,\]
  where $n$ is the absolute value of the image of $[i_2,i_2]$ by the equivalence
  $\pi_3(\Sn2)\stackrel\sim\longrightarrow\Z$.
\end{corollary}




\end{document}